\documentclass[a4paper,12pt,oneside,reqno]{amsart}
\usepackage{amssymb,amsfonts,amsmath,amsthm,amscd, bbm}
\usepackage{graphicx, color}
\usepackage{mathrsfs}
\usepackage{soul}
\usepackage [pdfstartview=FitH]{hyperref}
\newcommand{\beq}{\begin{equation}} 
\newcommand{\eeq}{\end{equation}} 
\newcommand{\bea}{\begin{aligned}}
\newcommand{\eea}{\end{aligned}}
\newcommand{\bdm}{\begin{displaymath}}
\newcommand{\edm}{\end{displaymath}}
\newcommand{\barr}{\begin{array}}
\newcommand{\earr}{\end{array}}
\newcommand{\ben}{\begin{enumerate}}
\newcommand{\een}{\end{enumerate}}
\newcommand{\bde}{\begin{description}}
\newcommand{\ede}{\end{description}}

\newcommand{\K}{\mathcal{K}}

\newtheorem{teor}{Theorem}[section]
 
\newtheorem{lem}[teor]{Lemma}

\newtheorem{rem}[teor]{Remark}
 
\newcommand{\R}{\mathbb{R}}

\newcommand{\N}{\mathbb{N}}

\newcommand{\PP}{\mathbb{P}}
\newcommand{\E}{{\mathbb{E}}}
\newcommand{\EE}{{\sf{E}}}

\newcommand{\e}{\epsilon}

\newcommand{\psu}{\mathcal{M}_1^+(\R)}
\newcommand{\ps}[1]{\mathcal{M}_1^+(\R^{#1})}

\newcommand{\Phib}{\boldsymbol \Phi}

\newcommand{\mK}{\mathcal{K}}

\newcommand{\Pqm}{P^{\, q,\bf m}}

\newcommand{\Pqmb}{P^{\, \bar q,{\bf \bar m}}}

\newcommand{\Eqm}{\EE^{\, q,\bf m}}

\newcommand{\Eqmb}{\EE^{\, \bar q,{\bf \bar m}}}

\newcommand{\Hqmb}{H^{\, \bar q,{\bf \bar m}}}

\newcommand{\Hqm}{H^{\, q,\bf m}}

\newcommand{\xj}{{\bf x}_j}

\begin{document}

\title[GREM approximation of TAP free energies]{On the GREM approximation of TAP free energies}

\author[G. Sebastiani]{Giulia Sebastiani}
\address{Giulia Sebastiani \\ J.W. Goethe-Universit\"at Frankfurt, Germany.}
\email{sebastia@math.uni-frankfurt.de}

\author[M. A. Schmidt]{Marius Alexander Schmidt}
\address{M. A. Schmidt \\ J.W. Goethe-Universit\"at Frankfurt, Germany.}
\email{M.Schmidt@mathematik.uni-frankfurt.de}

\thanks{We are grateful to Nicola Kistler for leading the way in this line of research.}

\subjclass[2000]{60J80, 60G70, 82B44} \keywords{Mean Field Spin Glasses, Large Deviations, Gibbs-Boltzmann and Parisi Variational Principles, Random Energy Models} 

\date{\today}

\begin{abstract}
We establish both a Boltzmann-Gibbs principle and a Parisi formula for the limiting free energy of an abstract GREM (Generalized Random Energy Model) which provides an approximation of the TAP (Thouless-Anderson-Palmer) free energies associated to the Sherrington-Kirkpatrick (SK) model.
\end{abstract}
\maketitle
\date{\today}
\maketitle
\section{Introduction}

Introduced by Derrida \cite{D, Dgrem} in the 80s, the Generalized Random Energy Models (GREMs) serve as simplified, abstract models that are completely solvable and have significantly contributed to our comprehension of specific facets within the Parisi theory for mean field spin glasses \cite{MPV}.

The simplest of Derrida’s models is the Random Energy Model (REM) \cite{D}, whose Hamiltonian at volume $N\in\N$ is given by a vector $\{H_N(\sigma)\}_{\sigma\in\{-1,+1\}^N}$ of independent Gaussian random variables with variance $N$. The GREMs \cite{Dgrem} generalize the REM, featuring hierarchically correlated Gaussian fields that inherently exhibit the ultrametricity predicted by the Parisi theory. 

In \cite{REMTAP} we related the REM and the Thouless-Anderson-Palmer (TAP) free energies \cite{TAP} for the Sherrington-Kirkpatrick (SK) model \cite{sk}, the prototypical mean field spin glass. The linkage is based on the observation that the free energy of {\it TAP-solutions} for the SK can be expressed as a {\textit{highly nonlinear}} functional of the {\textit{local fields}}  (see \eqref{eff_fields} and \eqref{TAP_SK_function} below). Through a {\textit{replacement}} of these fields we obtain abstract models involving only the alleged geometrical properties of the {\it relevant} TAP-solutions, that are solvable using Boltzmann formalism through a classical, Sanov-type large deviation analysis.

To see how this unfolds and for the reader's convenience, we briefly review the setting given in Section 2 of \cite{REMTAP}. The starting point is a non-rigorous result \cite{TAP} of Thouless, Anderson, and Palmer, which expresses the free energy associated to the SK Hamiltonian (by means of a %(problematic) 
diagrammatic expansion of the partition function) as a constrained maximum of a certain disorder-dependent function, the so-called TAP free energy \eqref{TAP_fe_no_cav}. This function is defined on the space of possible magnetizations of the spins, given as $[-1,1]^N$ where $N$ is the system size.

With $\beta \in \R^+$ the inverse temperature parameter, $h\in \R$ an external magnetic field and $J_{ij}$ i.i.d. standard Gaussians issued on some probability space $(\Omega, \mathcal{F}, \PP)$, namely the quenched disorder, the SK Hamiltonian at volume $N$ is given by the Gaussian field $\sigma\in\{-1,+1\}^N\mapsto (2N)^{-\frac{1}{2}}\beta\sum_{i \neq j} J_{ij}\sigma_i\sigma_j+h\sum_{i=1}^N \sigma_i$. The TAP free energy in ${\bf m}=(m_i)_{i=1}^N$ can be written, shortening $\bar J_{ij}:=(J_{ij}+J_{ji})/\sqrt{2}$, as
\beq \label{TAP_fe_no_cav}
\frac{\beta}{\sqrt{N}} \sum\limits_{i < j } \bar J_{ij} m_i m_j +h\sum\limits_{i=1}^{N} m_i -\sum\limits_{i=1}^{N}I(m_i) +\frac{\beta^2N}{4}(1-\frac{1}{N}\sum\limits_{i=1}^{N} m_i^2)^2
\eeq
where $I$ is the classic coin tossing rate function $I(m):=\frac{1+m}{2}\log(1+m)+\frac{1-m}{2}\log(1-m)$, which can also be expressed as 
\beq \label{alternative_coin_tossing}
I(m)=m\tanh^{-1}(m)-\log\cosh\tanh^{-1}(m).
\eeq
The TAP formulation is notably compelling, providing a physically meaningful representation consisting of three discernible components: internal energy, entropic contribution and Onsager correction. Hence, the construction and study of simplified models incorporating the structural elements of the TAP formulation gives a natural avenue to enhance the framework that ideally bridges the classical Boltzmann-Gibbs theory with the powerful and intricate, yet as of now still not fully understood, Parisi theory for mean field spin glasses.\\

A quick computation shows that the critical points of the TAP free energy \eqref{TAP_fe_no_cav} must be solutions of the {\textit{TAP equations}} 
\beq \label{TAP_eq} m_{i} = \tanh[\, h + \beta h_i({\bf m})\,], \qquad i=1, \dots, N,\eeq
where 
\beq \label{eff_fields}
h_i({\bf m}):=\frac{1}{\sqrt{N}} \sum_{j\neq i} \bar J_{ij}m_j - \beta(1-\frac{1}{N}\sum_{j=1}^Nm^2_j)m_i.
\eeq

\begin{rem}
For an analysis of the TAP equations, see e.g. \cite{FMM, GSS, T11}.
A proof that a constrained maximum of the TAP free energy yields the exact SK free energy can be found in \cite[Theorem 1]{CP}; the formulation of the constraint however strongly relies on the celebrated Parisi formula for the SK free energy. For more on the TAP-Plefka approximation within Guerra's interpolation scheme \cite{g}, see \cite{CPS} and references therein. 
\end{rem}

If we assume to have a solution to \eqref{TAP_eq}, say ${\bf m}^{\alpha}\in[-1,1]^N$, and evaluate the TAP free energy \eqref{TAP_fe_no_cav} in  ${\bf m}^{\alpha}$; we find it to be equal to $\text{TAP}_N({\bf m}^{\alpha})$, where 
\beq \label{TAP_SK_function} \bea 
\text{TAP}_N({\bf m}) & := \frac{\beta^2}{4}N\left[\,1-\left(\frac{1}{N}\sum_{j=1}^N\tanh^2(h+\beta h_j({\bf m}))\right)^{\,2}\,\right] +\\
& \quad + \sum_{i=1}^N \left(\,- \frac{\beta}{2}h_i({\bf m})\tanh(h+\beta h_i({\bf m})) + \log\cosh(h+\beta h_i({\bf m}))  \, \right).
\eea \eeq
This is easily seen, as follows. If ${\bf m}^{\alpha}$ is a solution to \eqref{TAP_eq}, then the equality
$$\beta h_i({\bf m}^\alpha)=\tanh^{-1}(m_{i}^\alpha) -  h$$ 
must hold for every $i=1,\dots, N$. Using the definition \eqref{eff_fields}, the latter implies
\beq \label{Jij_substitution}\frac{\beta}{2}\frac{1}{\sqrt{N}} \sum_{j\neq i} \bar J_{ij}m_j^\alpha = \frac{1}{2}\tanh^{-1}(m^\alpha_i) -\frac{h}{2} + \frac{\beta^2}{2}(1-q_N^{\,\alpha} )m^\alpha_i, \eeq
where we shortened 
$q_N^{\,\alpha} := \frac{1}{N}\sum_{j=1}^{N}(m^\alpha_j)^2$.
\\
Consider now the TAP free energy \eqref{TAP_fe_no_cav} evaluated in ${\bf m}={\bf m}^\alpha$, which can be rewritten as
$$\frac{\beta^2N}{4}(1-q_N^{\,\alpha})^2 + \sum\limits_{i=1}^{N} \left[\left( \frac{\beta}{2}\frac{1}{\sqrt{N}}\sum\limits_{j\neq i} \bar J_{ij} m^\alpha_j + h\right) m^\alpha_i - I(m^\alpha_i)\right].$$ 
Plugging \eqref{Jij_substitution} into this expression, we get
\beq \bea 
&\frac{\beta^2N}{4}(1-q_N^{\,\alpha})^2+ \frac{\beta^2N}{2}(1-q_N^{\,\alpha})q_N^{\,\alpha} +\sum\limits_{i=1}^{N} \left[\frac{1}{2}m^\alpha_i\tanh^{-1}(m^\alpha_i) +\frac{h}{2}m^\alpha_i - I(m_i^\alpha) \right] \\
&= \frac{\beta^2N}{4}[1-(q_N^{\,\alpha})^2]+\sum\limits_{i=1}^{N} \left[ -\frac{1}{2}m_i^\alpha\tanh^{-1}(m^\alpha_i) + \frac{h}{2}m_i^\alpha + \log\cosh\tanh^{-1}(m^\alpha_i) \right],\\
\eea \eeq
where we used the representation \eqref{alternative_coin_tossing} for the function $I$. 

Expressing each $m^\alpha_i$ as $\tanh(h+\beta h_i({\bf m}^{\alpha}))$, specifically $q_N^\alpha=\frac{1}{N}\sum_{i=1}^N \tanh(h+\beta h_i({\bf m}^{\alpha}))$, the latter becomes 
\beq \label{TAP_SK_function_2} \bea
& \frac{\beta^2}{4}N\left[\,1-\left(\frac{1}{N}\sum_{i=1}^N\tanh^2(h+\beta h_i({\bf m}^\alpha))\right)^{\,2}\,\right] +\\
& \qquad\qquad + \sum_{i=1}^N \left(\,- \frac{\beta}{2}h_i({\bf m}^\alpha)\tanh(h+\beta h_i({\bf m}^\alpha)) + \log\cosh(h+\beta h_i({\bf m}^\alpha))  \, \right),
\eea \eeq
which is exactly $\text{TAP}_N({\bf m}^\alpha)$ as defined in \eqref{TAP_SK_function}.
\vskip0.5cm

The {\textit{REM approximation of TAP free energies}} \cite{REMTAP} is achieved by replacing the local fields ${h_i({\bf m}^\alpha)}$ with a set of REM-like variables, specifically independent standard Gaussians ${X_{\alpha,i}}$, where $i=1,\dots,N$ and $\alpha=1\dots 2^{N}$. 

Positioned as the second natural step beyond the REM-approximation, this paper explores the \textit{GREM approximation of TAP free energies}, wherein the local fields are substituted with GREM-like variables (see \eqref{GREM_vectors} below). 
For the GREM approximation of TAP free energies, we can establish an orthodox Boltzmann-Gibbs principle and introduce a dual formula (given in Theorem \ref{parisi_grem_tap_teor} below) for the free energy, akin to the Parisi formula, which closely resembles the one given by the Parisi theory for models within the {\textit{1-step replica symmetry breaking}} approximation. The highlight is that the inherent non-linearities of \eqref{TAP_SK_function} bring the Parisi formulation for our GREM-like models into closer correspondence with that for the free energy of the SK, presenting a departure from conventional GREM formulations (for more on this, see Remark \eqref{remark_remtap_1rsb} below).
Our original result, Theorem \ref{parisi_grem_tap_teor} below, suggests that the correction to the Parisi formula, caused by the inherent nonlinearities, remains unchanged if we perform a \textit{GREM-replacement} instead of a \textit{REM-replacement}; that is, if we admit more than one level of hierarchy. Specifically, this correction does not vary for models within the \textit{$n$-steps replica symmetry breaking} approximation, where $n\in\N$ is the number of observable levels of hierarchy.\\

In the next section, we provide the necessary definitions and present our results, which we then prove in the subsequent two sections.

\section{Definitions and results}

Given $\Theta \in \R^+$ and $n\in \N$; we investigate, in an abstract setting, the following GREM-approximation of the TAP free energy \eqref{TAP_SK_function}
\beq \label{GREM_TAP_free_energy}
\lim_{N\to\infty} \frac{1}{N} \log \sum_{\alpha=1}^{2^{\Theta N}} \exp N {\bf \Phi}_{\text{TAP}}[ L_{N,\alpha} ]
\eeq
where, given positive numbers $\{\gamma_1,\dots, \gamma_n\}\,:\,\sum_{k =1}^n \gamma_k = 1$,
each configuration\\ $\alpha\in\{1,\dots, 2^{\Theta N}\}$ is identified with an $n$-tuple 
$$\alpha\equiv(\alpha_1,\dots,\alpha_n): \quad \alpha_k\in\{1,\dots, 2^{\gamma_k \Theta N}\}, \qquad k\leq n$$ 
and 
$$L_{N,\alpha}:=\frac{1}{N}\sum\limits_{i=1}^N \delta_{X_{\alpha,i}}$$
are the empirical measures associated to sequences $\{X_{\alpha,i}\}$ of random vectors
\beq \label{GREM_vectors}
X_{\alpha,i}=(X^{(1)}_{\alpha_1,i}, X^{(2)}_{\alpha_1, \alpha_2,i}, \dots,  X^{(n-1)}_{\alpha_1, \alpha_2\dots, \alpha_{n-1},i},   X^{(n)}_{\alpha,i})\in \R^n,
\eeq
with independent entries, being each $\{ X^{(k)}_{\alpha_1,\dots, \alpha_{k},i}\}_{\alpha_1,\dots, \alpha_{k},i}$ i.i.d. on $\R$.\footnote{We stress that $\R$ can be replaced with any Polish space, without this changing the validity of our proofs. We choose in this paper to just refer to the case in which this space is the real line for the sake of simplicity and with the aim of improving readability. For the {\textit{couple}} of topological considerations necessary to present the proofs in the case of variables defined on a generic Polish space, we refer the interested reader to previous works \cite{BK1, BK2, REMTAP}.}

Indicating with $\mathcal{M}^+_1(\R^n)$ the space of all regular Borel probability measures on $\R^n$ endowed with the topology of weak convergence of measures; the functional 
$${\bf \Phi}_{\text{TAP}}: \mathcal{M}^+_1(\R^n)\to \R$$ is continous and given by
\beq \label{TAP_SK_functional} \bea 
{\bf \Phi}_{\text{TAP}}[\nu] =& \Phi_{\text{SK}}\left( \, \int \varphi_{\text{SK}}\left(\sum_{k=1}^n x_k\right) \nu(d{\bf x}\,) \right)+\int f_{\text{SK}}\left(\sum_{k=1}^n x_k\right) \nu(d{\bf x})
\eea \eeq 
where, with ${\bf x}=(x_1,\dots, x_n)$ a generic point in $\R^n$  
\beq \label{specific_SK} \bea 
\Phi_{\text{SK}}:\R\to\R \,, \quad &\Phi_{\text{SK}}(x) := \frac{\beta^2}{4} (1-x^2) \,;\\
\varphi_{\text{SK}}:\R\to\R \,, \quad &\varphi_{\text{SK}}(x) := \tanh^2(h+\beta x)\,;\\
f_{\text{SK}}:\R\to\R \,, \quad & f_{\text{SK}}(x) := -\frac{\beta}{2}x\tanh(h+\beta x)+ \log\cosh(h+\beta x)\,.
\eea \eeq
It is readily checked that $\Phib_{\text{TAP}}[L_{N,\alpha}]$ 
is obtained taking \eqref{TAP_SK_function} for ${\bf m}={\bf m}^{\alpha}$ and substituting the fields $h_i({\bf m}^{\alpha})$ with the GREM-like variables 
\beq \label{grem_fields}
\sum_{k=1}^n X^{(k)}_{\alpha_1,\dots, \alpha_k,i}.
\eeq
As established in \cite{BK2}, if to given $\nu\in\ps{n}$ we indicate with $\nu^{(k)}\in\ps{k}$ its marginal on the first $k\in\{1,\dots,n\}$ coordinates, and consider the relative entropy 
\beq \label{def_rel_ent} H(\nu^{(k)}\mid \mu^{(k)}) := \begin{cases} \int \log \frac{d\nu^{(k)}}{d\mu^{(k)}} \, d\nu^{(k)} & \quad \text{if}\,\, \nu^{(k)}<<\mu^{(k)}\\
+\infty & \quad \text{otherwise}\,\,,
\end{cases}\eeq 
where 
$$\mu= \bigotimes_{k=1}^{n} \mu_k \in \mathcal{M}^+_1(\R^n)$$ 
is the common law of the vectors \eqref{GREM_vectors}; the following can be proved within a classical large deviation treatment via second moment method and Sanov's theorem (see Section \ref{section_proof_gibbs} below).

\begin{teor}[Boltzmann-Gibbs principle]\label{gibbs}
Let ${\bf \Phi}:  \mathcal{M}^+_1(\R^n)\to \R$ be a continous functional and consider
$$f_N := \frac{1}{N} \log \sum_{\alpha=1}^{2^{\Theta N}} \exp N {\bf \Phi}[ L_{N,\alpha} ].$$
Then the limit $\lim_{N\to\infty} f_N$ exists, is non-random and given by
\beq \label{GREM_TAP_free_energy_teor}
f := \sup_{\nu \in \mathcal K} \quad {\bf \Phi}[\nu] - H(\nu\mid \mu) + \Theta\log2  \\
\eeq
with
\beq \label{set_K}
\mathcal K := \left\{ \nu \in  \mathcal{M}^+_1(\R^n): \quad H(\nu^{(k)} \mid \mu^{(k)} ) \leq \, \Theta\log2 \sum_{j=1}^k \gamma_j \quad \forall k \in \{1,\cdots,n\} \right\}.
\eeq
\end{teor}

Given bounded and continous functions $\varphi,f: \R^n \to \R$, as well as a differentiable function $\Phi:\R\to\R$; our result Theorem \ref{parisi_grem_tap_teor} below, concerns a specific case in which the functional $\Phib$ is of the form\footnote{Notice that choosing $\varphi({\bf x}) \equiv \varphi_{\text{SK}}(\sum_{k=1}^n x_k)$, $f({\bf x}) \equiv f_{\text{SK}}(\sum_{k=1}^n x_k)$ and $\Phi(x)\equiv \Phi_{\text{SK}}(x)$ the generic functional \eqref{specific_functional} becomes the specific \textit{TAP-functional} \eqref{TAP_SK_functional}.}
\beq \label{specific_functional} {\bf \Phi}[\nu] \equiv \Phi( \int \varphi d\nu) +  \int f d\nu \eeq
and extends to $n\geq 2$ the {\textit{duality}} between the Boltzmann-Gibbs infinite dimensional principle of Theorem \ref{gibbs} and a finite dimensional, Parisi principle. 
Precisely, we prove (see Section \ref{section_proof_parisi} below) that, considering the following recursively constructed functions
\beq \label{recursive_Par} \bea 
\Pqm_n(x_1, \dots, x_n) & = \Phi'(q)\varphi(x_1, \dots, x_n) + f(x_1, \dots, x_n), \\
\Pqm_{j-1} (x_1, \dots, x_{j-1}) & = \frac{1}{m_j} \log \int \exp{m_j \Pqm_j(x_1,\dots, x_j)} \mu_j(dx_j) \qquad \forall j\leq n
\eea \eeq
parametrized by $(q, {\bf m})$ such that $q\in [\inf \varphi, \, \sup \varphi]$ and 
$${\bf m} = (m_1, \dots, m_n)\in \Delta :=\{ {\bf m}\in[0,1]^n: 0 < m_1 \leq m_2 \leq \dots \leq m_n \leq 1\};$$ 
the following holds:

\begin{teor}[Parisi principle]\label{parisi_grem_tap_teor}
If $\varphi,f: \R^n \to \R$ are bounded; $\Phi:\R\to\R$ is a twice differentiable function with $\Phi''(x)< 0$ for all $x\in[\inf \varphi, \, \sup \varphi]$ and 
\beq \label{TAP_SK_functional_teor}
{\bf \Phi}[\nu] := \Phi( \int \varphi d\nu) +  \int f d\nu; 
\eeq
then
$$\sup_{\nu \in \mathcal K} \, {\bf \Phi}[\nu] - H(\nu\mid \mu) = \inf_{(q,{\bf m})\in[\inf \varphi,\, \sup \varphi]\times \Delta} P(q,{\bf m}),$$
\beq \label{parisi_grem_tap}
P(q,{\bf m}) := \Phi(q)-q\Phi'(q) + \Pqm_0 + \sum_{k=1}^n \frac{\gamma_k \Theta\log2}{m_k} - \Theta\log2. 
\eeq
\end{teor}

\begin{rem} \label{remark_remtap_1rsb}
The case $n=1$ corresponds to the {\textit{REM-approximation}} of the TAP free energy which, as pointed out in the introduction, is the subject of study in \cite{REMTAP}.
The novelty of the present work and \cite{REMTAP} w.r.t to previous works (\cite{BK1, BK2}) is the nonlinearity of the functional \eqref{TAP_SK_functional_teor}. In fact, \eqref{TAP_SK_functional_teor} is continous, but not linear on $\mathcal{M}^+_1(\R^n)$; namely not of the form $\int \varphi d\nu$. 
Specifically in \eqref{parisi_grem_tap}, the nonlinearity causes the presence of the same correction term $\Phi(q)-q\Phi'(q)$, for any $n\in\N$. This term is zero in the classic / linear cases (e.g. for $\Phi$ the identity function) while reproducing, for the choices \eqref{specific_SK}, the unsettling quadratic (in the inverse temperature) appearing in the \textit{1RSB free energy} for the SK model. For more details on the comparison between the REM-approximation of TAP free energies and the 1RSB Parisi solution for the free energy of the SK, we refer the reader to Subsection 3.1 of \cite{REMTAP}. 
\end{rem}

In \cite{BK2} the authors focus on continous and linear functionals $\Phib$ and establish both the Boltzmann-Gibbs principle and the Parisi formula for the limiting free energy. Their proof of the Boltzmann-Gibbs principle, whose formulation is entirely analogous to that given in Theorem \ref{gibbs}, is rather powerful and does not require the linearity assumption. We refer to \cite[Theorem 2.3]{BK2} for a detailed proof, but give in Section \ref{section_proof_gibbs} a comprehensive outline for the reader's convenience. In the last Section \ref{section_proof_parisi} we prove Theorem \ref{parisi_grem_tap_teor}, namely we see how the finite dimensional Parisi principle encodes the optimal measure for the infinite dimensional Boltzmann-Gibbs principle.

\section{The Boltzmann-Gibbs principle: proof of Theorem \ref{gibbs}} \label{section_proof_gibbs}

Some infrastructure is necessary:  if a sequence of events $\left\{E_N\right\}_{N\in\N}\subset\mathcal{A}$ is such that there exists $\delta>0$ s.t. for $N$ large enough $\PP(E_N)\geq 1-e^{-\delta N}$ we say that $E_N$ holds with \textit{overwhelming probability} (o.p.). With $B_r(\nu)$ we denote the open ball centered in $\nu\in \mathcal{M}^+_1(\R^n)$ with radius $r>0$ in one of the standard metrics (e.g. Prokhorov's metric); a {\it{neighborhood}} or an {\it{open subset}} are conceived in the corresponding topology. This turns $\mathcal{M}^+_1(\R^{n})$ into a complete, separable metric space. 
The argument follows a second moment method performed on the counting variable
\beq \label{counting_GREM_tap_cav}
 M_N(U) := \# \{ \alpha\in \{1,\dots, 2^{\Theta N}\}: L_{N, \alpha} \in U \}\eeq
defined for every $U\subset \mathcal{M}^+_1(\R^n)$. Estimating its first moments on exponential scale through Sanov's Theorem, we can easily prove that for any $\e'>0$ 
\begin{itemize}
\item[$i)$] if $H(\nu \mid \mu) < +\infty$, then we can find an arbitrary small neighborhood $U$ of $\nu$:
\beq \label{Markov} M_N(U) \leq \exp{N[\Theta\log2 - H(\nu \mid \mu) +\e']}\,\, \text{with o.p.}\,;\eeq 
\item[$ii)$] if $H(\nu^{(k)} \mid \mu^{(k)}) > \Theta\log2 \sum_{j=1}^k \gamma_j $ for some $k\in\{1,\dots,n\}$, i.e. if $\nu\not\in \mathcal{K}$, then we can find an arbitrary small neighborhood $U$ of $\nu$:
\beq \label{marginal_cond} M_N(U) = 0 \,\, \text{with o.p.}\,.\eeq 
\end{itemize}
The statement $i)$ follows from a straightforward application of Markov's inequality; while in order to show $ii)$ we need to consider, given $\e'>0$ and a neighborhoods $U'$ of $\nu$, the subset
\beq U:=\{ \nu\in  \mathcal{M}^+_1(\R^n): \nu \in U', \nu^{(k)}\in U''\}\eeq 
where $U'' \subset  \mathcal{M}^+_1(\R^k)$ is such that
\beq \label{marginal_cond_1} \inf_{\nu'\in \text{cl}U''} H(\nu' \mid \mu^{(k)}) =  \inf_{\nu'\in U''} H(\nu' \mid \mu^{(k)}),\qquad \inf_{\nu'\in U''} H(\nu' \mid \mu^{(k)}) \geq \Theta \log2 \sum_{j=1}^k \gamma_j   + \eta.\eeq
for some $\eta>0$. Then, noticing that the marginal of $L_{N, \alpha}$ on the first $k$ coordinates only depend on the first $k$ entries of $\alpha$, to wit $ L_{N,\alpha}^{(k)} = N^{-1} \sum_{i\in[N]} \delta_{(X^1_{\alpha_1,i}, \dots, X^k_{\alpha_1,\dots,\alpha_k,i})}$; we get
\beq \bea 
\PP( \exists \alpha: L_{N,\alpha}\in U) &\leq \PP( \exists (\alpha_1,\dots, \alpha_k): L_{N,\alpha}^{(k)}\in U'') \leq 2^{\Theta \sum_{j=1}^k \gamma_j N} \PP(L_{N,1}^{(k)}\in U'').\\
\eea \eeq
The latter, together with Sanov's Theorem and \eqref{marginal_cond_1}, imply $ii)$. Now, as $\Phib$ is continous on $\mathcal{M}^+_1(\R^n)$ and $\mathcal{K}$ is compact, it is easy to see that $i)$, $ii)$, together with the Borel-Cantelli Lemma, imply the upper bound.
Specifically that there exists a finite covering of $\mathcal{K}$, say $\{ B_{r_l}(\nu_l)\}_{l=1}^M$, such that $\PP$-a.s.
\beq \label{upper_bound_GREM_TAP} \bea
\limsup\limits_{N\to \infty} f_{N}&\leq\limsup\limits_{N\to \infty} \frac{1}{N}\log\sum\limits_{l=1}^{M}\sum\limits_{\alpha: L_{N,\alpha}\in  B_{r_l}(\nu_l)}\exp \,N \Phib[L_{N,\alpha}]\\
&\leq \limsup\limits_{N\to \infty} \frac{1}{N}\log\sum\limits_{l=1}^{M}M_N(B_{r_l}(\nu_l))\exp \{\,N \Phib[\nu_l]+\e' N\}\\
&\leq\limsup\limits_{N\to \infty} \frac{1}{N}\log\sum\limits_{l=1}^{M}\exp \{\,N \Phib[\nu_l]+N\Theta\log2- NH(\nu_l\mid \mu)+2\e' N\}
\eea \eeq 
which, as each $\nu$ is in $\mathcal{K}$, implies the upper bound:
\beq \label{upper_grem_cav} \limsup_{N\to\infty} f_N \leq \sup_{\nu\in\mathcal{K}}\,\Phib[\nu]-H(\nu\mid\mu)+\Theta\log2 .\eeq
The key to the claim \eqref{gibbs}, i.e. to the lower bound, is now a concentration property of the counting variable \eqref{counting_GREM_tap_cav}. Specifically, by independence, for any $U\subset \mathcal{M}^+_1(\R^n)$ 
\beq \bea
\E \, M^2_N(U) &\leq [\E M_N(U)]^2+ \sum_{k=1}^n \sum_{\stackrel{\alpha, \alpha': \alpha_i = \alpha'_i \, \forall i\leq k}{ \alpha_{k+1}\neq \alpha'_{k+1}}} \PP( L_{N,\alpha}\in U, L_{N,\alpha'}\in U) 
\eea\eeq
so that a {\textit{relative}} version of Sanov Theorem (see \cite[Lemma 3.2]{BK2}) implies from the latter that if $H(\nu^{(k)} \mid \mu^{(k)}) < \Theta\log2 \sum_{j=1}^k \gamma_j$ for all $k\in\{1,\dots,n\}$, then for any small enough neighborhood $U$ of $\nu$:
\beq \mathbb{V}\text{ar} M_N(U)  \leq e^{-\delta N} [ \E M_N(U) ]^2 \quad\text{for some}\quad \delta > 0.\eeq
A classic application of Chebyshev inequality and the Borel-Cantelli Lemma now imply\footnote{For more details see the proof of \cite[Lemma 3.5]{BK2} and that of \cite[Proposition 3.2]{NHP}.}
\beq \label{concentration} \lim\limits_{N\to \infty}\frac{1}{N} \log M_N(U) =  \lim\limits_{N\to \infty}\frac{1}{N} \log \E M_N(U) = \Theta \log2 - H(\nu\mid\mu) \,.\eeq 
To see how \eqref{concentration} leads to the lower bound, first notice that $\mathcal{K}\subset  \mathcal{M}^+_1(\R^n)$ is compact and convex and $\nu\to {\bf \Phi}[\nu] - H(\nu\mid \mu)$ is upper-semicontinous and concave, therefore the Boltzmann-Gibbs principle \eqref{GREM_TAP_free_energy_teor} admits a unique solution. Indicating with $\nu^* \in \mathcal{K}$ the solution to \eqref{GREM_TAP_free_energy_teor}, consider the convex combination
\beq \nu_a := a\mu + (1-a)\nu^*\in\mathcal{K} \quad \text{for} \quad a\in(0,1). \eeq
One has that $\nu_a\to\nu^*$ weakly as $a\to0$, and $\Phib[\nu_a]\to\Phib[\nu^*]$, $H(\nu_a\mid\mu)\to H(\nu^*\mid\mu)$. Given $\e'>0$ we choose $a=a(\e')>0$, and a neighborhood $U$ of $\nu_a$:
\beq \Phib[\nu_a] - H(\nu_a\mid\mu) \geq \Phib[\nu^*] -  H(\nu^*\mid\mu) -\e',\eeq
\beq \Phib[\nu] \geq \Phib[\nu_a]-\e' \quad \forall \nu \in U.  \eeq
Then 
\beq \bea 
f_N &\geq \frac{1}{N}\log \sum_{\alpha: L_{N,\alpha} \in U} \exp{ N\Phib[L_{N,\alpha}]}\geq\Phib[\nu_a] - \e' + \frac{1}{N}\log M_N(U)
\eea \eeq
so that 
\beq \bea 
\liminf\limits_{N\to \infty} f_N &\geq \Phib[\nu_a] - H(\nu_a\mid\mu)+ \Theta \log2-\e' \geq \Phib[\nu^*] - H(\nu^*\mid\mu) + \Theta \log2-2\e',
\eea \eeq
which gives the lower bound. This, together with \eqref{upper_grem_cav}, proves Theorem \ref{gibbs}.

\section{The Parisi principle: proof of Theorem \ref{parisi_grem_tap_teor}} \label{section_proof_parisi}

In this Section we prove Theorem \ref{parisi_grem_tap_teor}. First, we show that the Parisi function \eqref{parisi_grem_tap} has a minum point on $[\inf \varphi, \sup\varphi]\times\Delta$. To this aim, notice that applying Jensen's inequality and Fubini's theorem $n$ times we get

$$\bea P_{0}^{\,q,{\bf m}} & = \frac{1}{m_1} \log \int \exp{m_1 P^{\,q,{\bf m}}_1(x_1)} \mu_1(dx_1) \geq \int P_1^{(q,{\bf m})}(x_1) \mu_1(dx_1)\\
& =  \int \frac{1}{m_2} \log \left( \int \exp{m_2 P_2(x_1, x_2)} \mu_2(dx_2) \right) \,  \mu_1(dx_1)\geq \int P_2^{\,q,{\bf m}}(x_1, x_2) \mu_1\otimes\mu_2(dx_1,dx_2) \\
& \geq \int P^{\,q,{\bf m}}_n({\bf x}) \mu(d{\bf x}) = \int \,   \Phi'(q) \varphi({\bf x}) + f({\bf x}) \, \mu(d{\bf x});\\
\eea$$
so that
$$ \bea 
P(q,{\bf m}) &= \Phi(q)-q\Phi'(q) + \Pqm_0 + \sum_{k=1}^n \frac{\gamma_k \Theta\log2}{m_k} - \Theta\log2\\
& \geq  \Phi(q)-q\Phi'(q) +  \sup_{{\bf x}\in\R^n}\mid \Phi'(q)\varphi({\bf x})+f({\bf x}) \mid + \sum_{k=1}^n \frac{\gamma_k \Theta\log2}{m_k} - \Theta\log2,
\eea $$
which implies that $P(q,{\bf m}) \to +\infty$ as $(q,{\bf m})$ tends in $[\inf \varphi, \sup\varphi]\times\Delta\subset \R^{n+1}$ to any point $(q',0, m'_2, \cdots, m'_n)$. The infimum of $P$ on $[\inf \varphi, \sup\varphi]\times\Delta$, where we recall $\Delta = \{{\bf m}=(m_1,\dots,m_n)\in[0,1]^n: 0< m_1\leq\dots\leq m_n\leq 1\}$, must be therefore equal to the infimum of $P$ on $[\inf \varphi, \sup\varphi]\times \{{\bf m}\in\Delta: m_1\geq \epsilon\}$, for some $\epsilon>0$.\\
Specifically, as $P$ achieves minimum on the compact set $[\inf \varphi, \sup\varphi]\times \{{\bf m}\in\Delta: m_1\geq \epsilon\}$, being there continous; it must achieve minimum on $[\inf \varphi, \sup \varphi]\times\Delta$. To wit
$$\text{argmin}_{[\inf \varphi, \sup\varphi]\times \Delta}P \neq \emptyset,$$ 
and we can state the following Lemma, whose simple proof shows how the Parisi functional enables  finding the solution to the Boltzmann-Gibbs principle.

\begin{lem} \label{lem_derivatives}
Let $\bar q\in \R, \bar{{\bf m}}=(\bar m_j)_{j=1}^n\in \R^n$ be s.t. $(\bar q, {\bf \bar m})\in \text{argmin}_{[\inf \varphi, \sup\varphi]\times \Delta}P$, and consider the probability measure $G=G^{\bar q, {\bar {\bf m}}}\in \mathcal{M}^+_1(\R^n)$ whose Radon-Nikodym derivative w.r.t. $\mu=\bigotimes_{j=1}^n \mu_j $ is
\beq \label{radon_candid_n}
\frac{dG}{d\mu}(x_1,\dots,x_n):=\prod_{j=1}^n \frac{e^{ \bar m_{j}\Pqmb_j(x_1,\dots, x_j) }}{\int e^{ \bar m_{j} \Pqmb_j(x_1,\dots, x_{j-1},x'_j)} \mu_j(dx'_j)}.
\eeq
Then $G\in \mathcal{K}$ and
\beq \label{upper_bound}
\Phib[\nu]- H(\nu\mid\mu) \leq \Phib[G]-  H(G\mid\mu)
\eeq
holds $\forall \nu \in \mathcal{K}$.
\vskip0.2cm
\end{lem}

\begin{proof}

For ${\bf x}=(x_1, \dots, x_n)$ a generic point in $\R^n$ and every $j\in\{1,\cdots, n\}$, we write $\xj := (x_1,\dots, x_j)$ for its projection on the first $j$ coordinates and shorten
$$P_j(\xj) := P_j^{\, \bar q, {\bar{\bf m}}}(\xj).$$
We also set for every couple $(q, {\bf m}) \in [\inf \varphi, \sup \varphi] \times \Delta$ and given a $\mu^{(j)}-$measurable function $v:S^j\to \R$

\beq \label{exp_marginal_j} \bea 
\Eqm_j( v )& :=  \int v(\xj) 
\frac{e^{ m_1\Pqm_1(x_1) }}{\int e^{ m_1 \Pqm_1(x'_1)} \mu_1(dx_1')  } \dots \frac{e^{ m_j\Pqm_j(\xj) }}{\int e^{ m_j \Pqm_j({\bf x}_{j-1}, x_j')} \mu_j(dx'_j)} \mu^{(j)}(d{\bf x}_j)\\
\eea \eeq

and 

\beq \label{rel_en_condjm1_onj} 
\bea
\Hqm_j({\bf x}_{j-1}) & :=  \int \log 
\left(\frac{e^{ m_j\Pqm_j(\xj) }}{\int e^{ m_j\Pqm_j({\bf x}_{j-1},x'_j) } \mu_j(dx'_j)}\right)  \,\frac{e^{ m_j\Pqm_j(\xj)  }}{\int e^{ m_j\Pqm_j({\bf x}_{j-1}, x'_j)} \mu_j(dx'_j)}\mu_j(dx_j),
\eea \eeq

with $\Hqm_1$ dependent only on $(q,{\bf m})$. It is readily checked that if $(q, {\bf m}) = (\bar q, {\bf \bar m})$ then $\Eqmb_j$ is the expectiation w.r.t the marginal $G^{(j)}$ of the measure $G$, and that therefore the relation
\beq \label{relation_E_Gj}
\Eqmb_j( v ) = \int v(\xj) \, G^{(j)}(d\xj)
\eeq
holds for every $j\in\{1,\cdots, n\}$ and every $\mu^{(j)}-$measurable function $v$. Upon closer inspection, it also becomes evident that 
\beq \label{relation_H_HG}
\Hqmb_1=H(G^{(1)}\mid \mu_1)
\eeq
and that if $j\geq 2$ and $(X_1,\dots, X_j)$ is a random vector distributed according to $G^{(j)}$, then $\Hqm_j({\bf x}_{j-1})$ is the Kullback-Leibler divergence between $\mu^{(j)}$ and the conditional probability distribution of $X_j$ given $X_1=x_1, \dots, X_{j-1}=x_{j-1}$.

This notation enables us to express partial derivatives of $P$ in a compact and convenient manner. Specifically, a straightforward computation shows that, if $\partial_q P$ is the partial derivative of $P$ w.r.t the variable $q$ and, for every $j\in\{1,\cdots, n\}$, $\partial_{m_j}P$ is the one w.r.t to the variable $m_j$; then 
\beq  \label{P_gradient}
\begin{cases}
\partial_q P(q,{\bf m}) =  \Phi''(q) \left[ \Eqm_n( \varphi ) - q \right] \\\\
\partial_{m_1} P(q,{\bf m}) =  \frac{1}{m_1^2}\left[   \Hqm_1 - \gamma_1\Theta\log2 \right]\\\\
\partial_{m_j} P(q,{\bf m}) = \frac{1}{m^2_j}\left[  \Eqm_{j-1}(\Hqm_j)- \gamma_j\Theta\log2 \right]\quad \forall j\in \{2,\cdots, n\}.
\end{cases}
\eeq
Through the first of these relations, the minimality of $\bar q$ for the univariate function $q \to P(q, {\bf \bar m})$ on  $[\inf \varphi, \sup \varphi]$, together with the fact that $\Phi''$ is strictly negative on $[\inf \varphi, \sup\varphi]$; imply
\beq \label{equality_q}\bar q  = \Eqmb_n( \varphi ) \stackrel{\eqref{relation_E_Gj} \, \text{with} \, j\equiv n}{=} \,\,\,  \int \varphi({\bf x}) G(d{\bf x}).\eeq

Similarly, the minimality of $\bar m_1\neq 0$ for the univariate function $m_1 \to P(\bar q, m_1, \bar{m}_2,\dots, \bar m_n)$ on $(0, \bar m_2]$
implies that $\partial_{m_1} P(\bar q,{\bf \bar m}) \leq 0$, i.e.
\beq \label{base_case_G_in_K} 
\gamma_1 \Theta \log 2 \geq \Hqmb_1 \,\,\, \stackrel{\eqref{relation_H_HG}}{=}  H(G^{(1)}\mid \mu_1) ;
\eeq
which amounts to say that the probability measure $G$ satisfies the {\textit{ first}} entropic condition appearing in the definition \eqref{set_K} of the set $\K$.\\
Given \eqref{base_case_G_in_K} as a base case, the fact that the couple $G$ actually satisfy {\textit{all}} the entropic conditions, i.e. that 
\beq \label{claim_induction_grem_tap} H(G^{(j)}\mid\mu^{(j)})\leq \Theta\log2 \sum_{k\in[j]}\gamma_k \eeq
for every $j\in\{1,\cdots, n\}$, and therefore $G\in\K$; can be easily proved by induction. To see how this comes about, consider $j\in\{2,\cdots,n\}$ and assume that for any $k<j$
\beq \label{inductive_hp_grem_tap} H(G^{(k)}\mid\mu^{(k)}) \leq \Theta\log2 \sum_{l\in[k]}\gamma_l \eeq
i.e. the inductive hypothesis.
Then consider
$$\text{min}_j := \min\{ k: \quad 1\leq k\leq j: \quad \bar m_k = \bar m_j \}, $$
and, setting $\bar m_0 := 0, \bar m_{n+1}:=1$, the interval $(\bar m_{\text{min}_j-1}, \bar m_{j+1} ]$. Notice that
one has $0\!\!<\!\!\bar m_1\!\leq\!\bar m_k \,\, \forall k\leq n$, so that the interval $(\bar m_{\text{min}_j-1}, \bar m_{j+1} ]$ is non-empty and we can therefore define
the real function $m\in  (\bar m_{\text{min}_j-1}, \bar m_{j+1} ] \to g_j(m)$ given by
\beq \label{block_function} g_j(m):= \begin{cases}
P(\bar q, m, \dots, m) & \text{if} \,\, \min_{j}=1, \, j=n \,;\\
P(\bar q,  \bar m_1, \dots, \bar m_{\text{min}_j-1},m, \dots, m) & \text{if} \,\, \min_{j}>1, \, j=n \,;\\
P(\bar q, m, \dots, m, \bar m_{j+1},\dots, \bar m_n) & \text{if} \,\, \min_{j}=1, \, j<n \,;\\
P(\bar q, \bar m_1, \dots, \bar m_{\text{min}_j-1},m, \dots, m, \bar m_{j+1},\dots, \bar m_n) & \text{if} \,\, \min_{j}>1, \, j<n\,.\\
\end{cases} \eeq
By construction, $g_j$ must have a minimum point in $m=\bar m_j \in (\bar m_{\text{min}_j-1}, \bar m_{j+1} ]$ which is a critical point for $g_j$ if $\bar m_j \neq \bar m_{j+1}$ while, possibly, such that $g'_j(\bar m_j) < 0$ if $\bar m_j = \bar m_{j+1}$.\\
As $g'_j(\bar m_j)=\sum\limits_{k=\text{min}_j}^j \partial_{m_k} P(\bar q, {\bf{\bar m}})$, this implies that it must hold $\sum\limits_{k=\text{min}_j}^j \partial_{m_k} P(\bar q, {\bf{\bar m}}) \leq 0$ in any case.\\
In order to include the case $\text{min}_j\equiv1$ we extend our notation setting $\Eqmb_0(\Hqmb_1)\equiv \Hqmb_1$ and $H(G^{(0)}\mid \mu^{(0)} )\equiv0$. Specifically, through the relations \eqref{P_gradient}, we obtain
\beq \label{derivative_to_rel_entr} \bea \sum\limits_{k=\text{min}_j}^j \partial_{m_k} P(\bar q, {\bf{\bar m}})&=\sum\limits_{k=\text{min}_j}^j  \frac{1}{\bar m^2_k}\left[ \Eqmb_{k-1}(\Hqmb_k) - \Theta\gamma_k\log2 \right]\\
&=\frac{1}{\bar m^2_j}\sum\limits_{k=\text{min}_j}^j\left[ \Eqmb_{k-1}(\Hqmb_k)- \Theta\gamma_k\log2 \right] \leq 0.\eea \eeq
The well known {\textit{chain-rule}} for Kullback-Leibler divergences is now basically everything we need. With our notation, it is equivalent to the fact that for each $k$ 
\beq \label{chain_rule_G} H(G^{(k)}\mid \mu^{(k)} ) = H(G^{(k-1)}\mid \mu^{(k-1)} ) + \Eqmb_{k-1}(\Hqmb_k),\eeq
which in turn implies that 
$$H(G^{(j)}\mid \mu^{(j)} )  = H(G^{(\text{min}_j-1)}\mid \mu^{(\text{min}_j-1)} )+ \sum_{k=\text{min}_j}^j  \Eqmb_{k-1}(\Hqmb_k).$$

Using the latter together with the last inequality in \eqref{derivative_to_rel_entr}, we get
\beq \bea & H(G^{(j)}\mid\mu^{(j)})\leq  H(G^{(\text{min}_j-1)}\mid \mu^{(\text{min}_j-1)} )+ \Theta\log2 \sum\limits_{k=\text{min}_j}^j\gamma_k 
\eea \eeq
so that the induction step follows using the inductive hypotesis \eqref{inductive_hp_grem_tap} for $k\equiv \text{min}_j-1 < j$. All in all, we proved that \eqref{claim_induction_grem_tap} holds for every $j\in\{1,\cdots n\}$, therefore that $G\in \mathcal{K}$.

\vskip0.5cm

We now prove \eqref{upper_bound}, i.e. our claim is that for any fixed  $\nu \in \mathcal K$
\beq \label{claim_fe}
\Phib[G ] -  H(G\mid \mu) -  \Phib[ \nu] + H(\nu \mid \mu) \geq 0.
\eeq
In order to lighten notation, we set
$$\bea 
& D_n :=\Phib[ G] -  H(G\mid \mu) -  \Phib[ \nu] +H(\nu\mid \mu);\\
& D_j :=  \int P_j\, dG^{(j)} - \frac{1}{\bar m_{j+1}} H(G^{(j)}\mid \mu^{(j)})-  \int P_j \,d\nu^{(j)} +  \frac{1}{\bar m_{j+1}} H( \nu^{(j)}\mid \mu^{(j)} ), \quad j\in \{1,\cdots,n-1\}.  
\eea $$
We will prove that $D_1 \geq 0$ and that $D_{k} \geq D_{k-1}$  for any $k\in \{2,\cdots, n\}$, which implies $D_n \geq 0$, namely the claim \eqref{claim_fe}. Both proofs follow a slight extension of the argument we used to show that $G\in\mK$.
We start by proving that $D_1 \geq 0$.\\
First notice that, considering the two possible cases $\bar m_1 = \bar m_2$ and $\bar m_1 < \bar m_2$, we have:
\begin{enumerate}
\item if $\bar m_1 = \bar m_2$ then 
\beq \label{D1_eq}
\bea D_1 = & \int P_1\, dG^{(1)} - \frac{1}{\bar m_{1}} H(G^{(1)}\mid \mu_1)- \int P_1 \,d\nu^{(1)}+  \frac{1}{\bar m_{1}} H\left( \nu^{(1)}\mid\mu_1\right)\,;
\eea\eeq
\item if $\bar m_1 < \bar m_2$ (as $\bar m_1 \neq 0$) the partial derivative $\partial_{m_1}P$ must vanish when evaluated in $(\bar q, {\bf{\bar m}})$, and this implies $H(G^{(1)}\mid \mu_1)=\gamma_1\Theta\log2$.\\ Therefore, since $\nu\in \K \Rightarrow H(\nu^{(1)}\mid \mu_1)\leq \gamma_1\Theta\log2$, it holds 
$$H(G^{(1)}\mid \mu_1)-  H(\nu^{(1)}\mid \mu_1)\geq 0,$$
and hence
\beq \label{D1_pre_ineq} \frac{1}{\bar m_1}\left[H(G^{(1)}\mid \mu_1) - H(\nu^{(1)}\mid \mu_1)\right]\geq  \frac{1}{\bar m_2}\left[H(G^{(1)}\mid \mu_1) -  H(\nu^{(1)}\mid \mu_1)\right]. \eeq
\end{enumerate}
It follows from \eqref{D1_eq}, \eqref{D1_pre_ineq} that in both cases it holds
\beq \label{ineq_D1}
\bea 
D_1 & \geq \int P_1\, dG^{(1)}  - \frac{1}{\bar m_{1}} H(G^{(1)}\mid \mu_1)- \int P_1 \,d\nu^{(1)} +  \frac{1}{\bar m_{1}} H(\nu^{(1)}\mid \mu_1).
\eea \eeq
Therefore, using the standard definition of relative entropy \eqref{def_rel_ent} to express $H(G^{(1)}\mid \mu_1)$, we get
\beq \bea
D_1 
&\geq \frac{1}{\bar m_1}\log \int \exp{\left( m_1 P_1 \right)} d\mu_1 - \int P_1 \,d\nu^{(1)} +  \frac{1}{\bar m_{1}} H(\nu^{(1)}\mid \mu_1)\\
& =  \frac{1}{\bar m_1}\left\{ H\left( \nu^{(1)}\mid \mu_1 \right) - \left[  \int \bar m_1 P_1 d\nu^{(1)}  - \log \int \exp{\left( m_1 P_1 \right)} d\mu_1 \right]\right\} \geq 0
\eea \eeq
where in the last step we used the Donsker and Varadhan variational formula for relative entropies 
$$H(\nu\mid \mu) = \sup_{f\in C_b} \{ \int f d\nu - \log \int e^f d\mu\}$$ (see \cite{DV,SYL} for details) and the fact that the hypotesis $\varphi, f \in C_b(\R^n)$ imply that $P_j \in C_b(S^j)$ for every $j\in\{1,\cdots, n\}$.\\
Now, we only need to prove that $D_{k} \geq D_{k-1}$  for any $k\in \{1,\cdots, n\}$.\\
Notice that if $K\in \{1,\cdots, n\}$ is the total number of different values assumed by the (positive and non-decreasing w.r.t the indices) entries of ${\bar {\bf m}}=(\bar m_1,\dots, \bar m_n)$; setting $\text{max}_0 := 0$, we can define for every $k\in \{1,\cdots, K\}$
\beq \text{max}_{k} := \max\{j:\quad  \text{max}_{k-1} + 1 \leq j \leq n, \,\,\, \bar m_j = \bar m_{\text{max}_{k-1}+1}\}. \eeq
To wit: for any $k\in \{1,\cdots, K\}$ the index $\text{max}_k$ is the last one for which the corresponding entry in ${\bf \bar m}$ assumes the $k$-th higher value in the whole vector: we will refer to the largest sections of the vector ${\bf \bar m}$ where all entries are equal as \textit{blocks}. By construction, $j=\text{max}_k$ for some $k$ iff \ $\bar m_{j} < \bar m_{j+1}$.\\
As in the case $1)$ above, if $\bar m_n = 1$ then 
\beq \label{first_ineq_Dn}
D_n \geq \Phib[ G ]  - \frac{1}{\bar m_n} H(G\mid \mu) -  \Phib[ \nu] + \frac{1}{\bar m_n} H(\nu \mid \mu);\\
\eeq 
holds with equality. If $\bar m_n < 1$ then, considering the real function $m\in (0, \bar m_{\text{max}_1+1}] \to g_{\text{max}_1}(m)$ defined by \eqref{block_function} for $j\equiv \text{max}_1$, to wit
$$g_{\text{max}_1}(m)= \begin{cases}
P(\bar q, m, \dots, m) & \text{if} \,\, \max_{1}=n \,,\\
P(\bar q, m, \dots, m, \bar m_{\text{max}_1+1}, \dots, \bar m_n)  & \text{if} \,\, \max_{1}<n\,;
\end{cases}$$
we notice that it must have a minimum point in $m=\bar m_1\in(0,\bar m_{max_1+1})$, specifically for $m\neq \bar m_{\text{max}_1+1}$ (recall that we have set $\bar m_{n+1}\equiv 1$). 
This implies that $\bar m_1$ must be a critical point of $g_{\text{max}_1}$, i.e. that $\sum\limits_{k=1}^{\text{max}_1} \partial_{m_k} P(\bar q, {\bf \bar m})=0$, which implies
\beq \label{1_max2_zero}
H(G^{(\text{max}_1)}\mid \mu^{(\text{max}_1)})=\Theta \log2 \sum\limits_{k=1}^{\text{max}_1} \gamma_k.
\eeq 
If now $K=1$, i.e. if ${\bf m}$ is made of one block and $\text{max}_1=n$, the latter is exactly $H(G\mid\mu)=\Theta\log2$. If instead $K\geq 2$, i.e. if $\bar {\bf m}$ is made of at least two blocks and $\text{max}_1<n$, we easily get, as above, that $\sum_{k=\text{max}_1+1}^{\text{max}_2} \partial_{m_k}P(\bar q, \bar {\bf m})=0$, which implies 
\beq \label{max12_zero}
\sum_{k=\text{max}_1+1}^{\text{max}_2}\Eqmb_{k-1}(\Hqmb_k)=\Theta \log2\sum_{k=\text{max}_1+1}^{\text{max}_2}\gamma_k.
\eeq
Through the chain-rule, \eqref{1_max2_zero} and \eqref{max12_zero}, we get
\beq \bea 
H(G\mid\mu) = & H(G^{(\text{max}_1)}\mid \mu^{(\text{max}_1)}) +  \sum_{k=\text{max}_1+1}^n \Eqmb_{k-1}(\Hqmb_k) \\
= & \Theta \log2 \sum\limits_{k=1}^{\text{max}_1} \gamma_k +  \Theta \log2 \sum\limits_{k=\text{max}_1+1}^{\text{max}_2} \gamma_k +  \sum_{k=\text{max}_2+1}^n \Eqmb_{k-1}(\Hqmb_k).
\eea \eeq
It should now be clear to the reader that by repeating the same reasoning $K$ times, since $\text{max}_K\equiv n$, we arrive at
\beq \bea 
H(G\mid\mu) = & \Theta \log2 \sum\limits_{j=0}^K \sum\limits_{k=\text{max}_j+1}^{\text{max}_{j+1}} \gamma_k =\Theta \log2.
\eea \eeq

Specifically, as before,
the inequality \eqref{first_ineq_Dn} holds also if $\bar m_n < 1$. From \eqref{chain_rule_G} and some elementary manipulations one finds 
\beq \bea 
H(G\mid\mu) & = H(G^{(n-1)}\mid \mu^{(n-1)}) + \Eqmb_{n-1}(\Hqmb_n) \\
& = H(G^{(n-1)}\mid \mu^{(n-1)}) + \bar m_n \int P_n \, dG  - \bar m_n \int P_{n-1} \, dG^{(n-1)} \\
& = H(G^{(n-1)}\mid \mu^{(n-1)}) + \bar m_n \Phi'(\bar q) \int \varphi \, dG + \bar m_n \int f \, dG   - \bar m_n \int P_{n-1} \, dG^{(n-1)}. \\
\eea \eeq
Through the latter, we get
\beq \bea 
\frac{1}{\bar m_n} H(G\mid \mu) =& \frac{1}{\bar m_n} H(G^{(n-1)}\mid\mu^{(n-1)}) + \Phi'(\bar q)\int \varphi dG+\int f dG - \int P_{n-1} dG^{(n-1)}
\eea \eeq
Plugging the latter into \eqref{first_ineq_Dn}, and recalling the explicit expression \eqref{specific_functional} for $\Phib$, we find
\beq \label{second_ineq_Dn} \bea 
D_n & \geq \Phi\left(  \int \varphi dG  \right) - \Phi\left( \int \varphi d\nu \right) -\Phi'(\bar q)\int \varphi \, dG - \int f d\nu \, +\\
&\qquad\qquad  + \frac{1}{\bar m_n} H(\nu\mid \mu) -\frac{1}{\bar m_n}H(G^{(n-1)}\mid\mu^{(n-1)}) +\int P_{n-1} \, dG^{(n-1)}  \\
&\geq\frac{1}{\bar m_n} H(\nu\mid\mu) - \Phi'(\bar q)\int \varphi d\nu  - \int f d\nu -\frac{1}{\bar m_n}H(G^{(n-1)}\mid \mu^{(n-1)}) + \int P_{n-1} \, dG^{(n-1)} \\
\eea \eeq 
the second inequality by concavity of $\Phi$, using that $\bar q = \int \varphi \, dG$.\\
Recall now that a chain-rule analogous to \eqref{chain_rule_G} holds for every measure $\nu\in \K$. Specifically, if $(X_1,\dots, X_n)$ is a random vector distributed according to some $\nu\in\ps{n}$ and ${\bf x}_{n-1}=(x_1,\dots, x_{n-1})\in \R^{n-1}$, 
indicating with  $K_\nu^{({\bf x}_{n-1})}\in \psu$ the regular conditional probability distribution of $X_n$ given $X_1=x_1, \dots, X_{n-1}=x_{n-1}$ 
the following holds
$$\bea
H(\nu\mid\mu)&
&= H(\nu^{(n-1)}\mid\mu^{(n-1)}) + \int \,\,H(K_\nu^{({\bf x}_{n-1})}\mid \mu_n) \,\, \nu^{(n-1)}(d{\bf x}_{n-1}).\eea $$
Using this formula we get
\beq \label{rel_en_from_n_to_n-1} \bea 
& H(\nu\mid\mu) - \int \,\, \bar m_n \left[ \Phi'(\bar q) \varphi({\bf x})+ f({\bf x})\right] \,\, \nu(d{\bf x}) + \bar m_n \int P_{n-1}({\bf x}_{n-1}) \,\, \nu^{(n-1)} (d{\bf x}_{n-1})\\
& =  H(\nu^{(n-1)}\mid\mu^{(n-1)}) + \int \,\,H(K_\nu^{({\bf x}_{n-1})}\mid \mu_n) \nu^{(n-1)}(d{\bf x}_{n-1}) \, +\\
& \quad  -\int \left\{ \int \bar m_n \left[ \Phi'(\bar q) \varphi({\bf x}) + f({\bf x}) \right] K_\nu^{({\bf x}_{n-1})}(dx_n) - \bar m_n P_{n-1}({\bf x}_{n-1}) \right\} \,\, \nu^{(n-1)}(d{\bf x}_{n-1}) \\
& = H(\nu^{(n-1)}\mid\mu^{(n-1)}) + \int \,\,H(K_\nu^{({\bf x}_{n-1})}\mid \mu_n) \nu^{(n-1)}(d{\bf x}_{n-1}) \, +\\
& \quad  -\int \left\{ \int \bar m_n \left[ \Phi'(\bar q) \varphi({\bf x}) + f({\bf x}) \right] K_\nu^{({\bf x}_{n-1})}(dx_n) -\log \int e^{\bar m_n P_n({\bf x})} \mu_n(dx_n) \right\} \nu^{(n-1)}(d{\bf x}_{n-1}).\\
\eea \eeq
Now, the Donsker and Varadhan variational formula for relative entropies implies that for every fixed ${\bf x}_{n-1} \in \R^{n-1}$
$$\bea 
H(K_\nu^{({\bf x}_{n-1})}\mid \mu_n) &
\geq  \int \bar m_n\left[ \Phi'(\bar q) \varphi({\bf x}) + f({\bf x}) \right] K_\nu^{({\bf x}_{n-1})}(dx_n) -  \log \int e^{ \bar m_n\left[ \Phi'(\bar q) \varphi({\bf x}) + f({\bf x}) \right] } \mu_n(dx_n) 
\eea$$
which, as $P_{n}({\bf x}) = \Phi'(\bar q) \varphi({\bf x}) + f({\bf x}) $, implies through \eqref{rel_en_from_n_to_n-1} that
\beq \bea 
\frac{1}{\bar m_n}H(\nu\mid\mu) - \int \,\, \left[ \Phi'(\bar q) \varphi+ f \right] \,\, d\nu \geq - \int P_{n-1} \,\, d\nu^{(n-1)}+ \frac{1}{\bar m_n}H(\nu^{(n-1)}\mid\mu^{(n-1)}).
\eea \eeq
Plugging the latter into \eqref{second_ineq_Dn} we find exactly $D_n \geq D_{n-1}$.\\
It is now enough to notice that the exact same argument (with $\Phi$ the identity function) implies $D_{k}\geq D_{k-1}$ for every $k\in\{1,\cdots,n-1\}$.\\ 
As we proved that $D_1 \geq 0$, the claim \eqref{claim_fe} follows and the proof of the Lemma is concluded. 
\end{proof}
It follows from Lemma \ref{lem_derivatives} that the unique solution to the Boltzmann-Gibbs principle must be given by $G^{\bar q, \bar {\bf m}}$, where $(\bar q, \bar {\bf m})$ is a minimum point for $P$ on $[\inf \varphi, \sup \varphi]\times \Delta$. As one easily checks that $\Phib[G]-H(G\mid \mu) = \Phib[G^{\bar q, \bar {\bf m}}]-H(G^{\bar q, \bar {\bf m}}\mid \mu) = P(\bar q, \bar {\bf m})$, Theorem \ref{parisi_grem_tap_teor} follows from Theorem \ref{gibbs} together with Lemma \ref{lem_derivatives}.


\begin{thebibliography}{1} 

\bibitem{ASS} Aizenman, Michael, Robert Sims, and Shannon L. Starr. {\it An Extended Variational Principle for the SK Spin-Glass Model.} Phys. Rev. B, 68:214403, (2003).

\bibitem{ASS2} Aizenman, Michael, Robert Sims, and Shannon L. Starr. {\it Mean-field spin glass models from the cavityROSt perspective.} Prospects in mathematical physics, volume 437 of Contemp. Math., pages 1–30.Amer. Math. Soc., Providence, RI, 2007. arXiv:math-ph/0607060.

\bibitem{AA} Arguin, Louis-Pierre, and Michael Aizenman. {\it On the structure of quasi-stationary competing particle systems.} The Annals of Probability 37.3 (2009): 1080-1113.

\bibitem{A} Bolthausen, Erwin, and Alain-Sol Sznitman. {\it On Ruelle's probability cascades and an abstract cavity method.} Communications in mathematical physics 197.2 (1998): 247-276.

\bibitem{B} Bolthausen, Erwin, and Alain-Sol Sznitman. {\it Ten lectures on random media.} Vol. 32. Springer Science \& Business Media, (2002).

\bibitem{BK1} Bolthausen, Erwin, and Nicola Kistler. {\it Universal structures in some mean field spin glasses and an application.} Journal of mathematical physics 49.12 (2008): 125205.

\bibitem{BK2} Bolthausen, Erwin, and Nicola Kistler. {\it A quenched large deviation principle and a Parisi formula for a Perceptron version of the GREM.} Probability in complex physical systems. Springer, Berlin, Heidelberg, (2012). 425-442.

\bibitem{P} Panchenko, Dmitry. {\it The Parisi ultrametricity conjecture.} Annals of Mathematics (2013): 383-393.

\bibitem{CP} Chen, Wei-Kuo and Dmitry Panchenko {\it On the TAP free energy in the mixed p-spin models.} Commun.
Math. Phys. 362 (2018): 219–252

\bibitem{CPS} Chen, Wei-Kuo, Dmitry Panchenko, and Eliran Subag. {\it The generalized TAP free energy II.} Communications in Mathematical Physics 381.1 (2021): 257-291.

\bibitem{CLR} Crisanti, Andrea, Luca Leuzzi, and Tommaso Rizzo. {\it Complexity in mean-field spin-glass models: Ising p-spin.} Physical Review B 71.9 (2005): 094202.

\bibitem{D} Derrida, Bernard. {\it Random Energy Model: An exactly solvable model of disordered systems.}, Phys. Rev. B 24 (1981)

\bibitem{Dgrem} Derrida, Bernard.  {\it A generalization of the random energy model which includes correlations between energies.} Journal de Physique Lettres 46.9 (1985): 401-407.

\bibitem{DV}  Donsker, Monroe D., and S.R. Srinivasa Varadhan. {\it Asymptotic evaluation of certain Markov process expectations for large time, I.} Communications on Pure and Applied Mathematics 28.1 (1975): 1-47.

\bibitem{WF}  Feller, William. {\it An introduction to probability theory and its applications.} Vols. I \& II, Wiley I 968 (1971).

\bibitem{FMM} Fan, Zhou, Song Mei and Andrea Montanari {\it TAP free energy, spin glasses and variational inference.}
Ann. Probab. 49(1): 1-45 (2021).

\bibitem{g} Guerra, Francesco.  {\textit Broken replica symmetry bounds in the mean field spin glass model.} Comm. in Math. Phys. Vol. 233 (2002).

\bibitem{G03} Guerra, Francesco. {\it About the cavity fields in mean field spin glass models.} (2003)
Available at arxiv:cond-mat/0307673.

\bibitem{GS} Gufler, Stephan,  Adrien Schertzer,  and Marius A. Schmidt. {\it On concavity of TAP free energy in the SK model}.  arXiv:2209.08985 (2022). 

\bibitem{GSS} Gufler, Stephan,  Adrien Schertzer,  and Marius A. Schmidt. {\it AMP algorithms and Stein's method: Understanding TAP equations with a new method}.  arXiv:2311.11924v1 (2023). 

\bibitem{REMTAP} Kistler, Nicola, Marius A. Schmidt and Giulia Sebastiani. {\it On the REM approximation of TAP free energies.} Journal of Physics A: Mathematical and Theoretical, Volume 56, Number 29 (2023)

\bibitem{SYL} Lee, Se Yoon. {\it Gibbs sampler and coordinate ascent variational inference: A set-theoretical review.} Communications in Statistics -Theory and Methods 51.6 (2022): 1549-1568.

\bibitem{MPV} Mezard, Marc, Giorgio Parisi, and Miguel Angel Virasoro. {\it Spin glass theory and beyond: An Introduction to the Replica Method and Its Applications}. Vol. 9. World Scientific Publishing Company, (1987).

\bibitem{Nish} Nishimori, Hidetoshi. {\it Statistical Physics of Spin Glasses and Information Processing: an Introduction.} Oxford; New York: Oxford University Press (2001).

\bibitem{plefka} Plefka, Timm. {\it Convergence condition of the TAP equation for the infinite-ranged Ising spin glass model.} Journal of Physics A: Mathematical and general 15.6 (1982): 1971.

\bibitem{PoWo} Polyanskiy, Yury, and Yihong Wu. {\it Lecture notes on information theory} (2019).

\bibitem{R} Ruelle, David. {\it A mathematical reformulation of Derrida's REM and GREM.}, Comm.Math.Phys. 108 (1987).

\bibitem{AR} Ruzmaikina, Anastasia, and Michael Aizenman. {\it Characterization of invariant measures at the leading edge for competing particle systems.} The Annals of Probability 33.1 (2005): 82-113.

%\bibitem{BRR} Bhattacharya, Rabi N., and R. Ranga Rao. {\it Normal approximation and asymptotic expansions}. Society for Industrial and Applied Mathematics, (2010).

\bibitem{NHP} Kistler, Nicola and Giulia Sebastiani. {\it On a nonhierarchical generalization of the Perceptron GREM.} Stochastic Processes and their Applications, Volume 163, Pages 1-24 (2023)

%\bibitem{kallenberg} Kallenberg, Olav. (2017). {\it Random measures, theory and applications}. Springer International Publishing.

\bibitem{sk} Sherrington, David, and Scott Kirkpatrick. {\it Solvable model of a spin-glass.} Physical review letters 35.26 (1975): 1792.

\bibitem{SGCH} Talagrand, Michel. {\it Spin Glasses: A Challenge for Mathematicians. Cavity and Mean Field Models.} Springer Verlag (2003).

\bibitem{T11} Talagrand, Michel {\it Mean field models for spin glasses.} Springer (2011).

\bibitem{TAP} Thouless, David J., Philip W. Anderson, and Robert G. Palmer. {\it Solution of 'solvable model of a spin glass'.} Philosophical Magazine 35.3 (1977): 593-601.

\end{thebibliography}
\end{document}